\documentclass{llncs}
\usepackage{cite}
\usepackage{xr}

\usepackage{makeidx}  

\usepackage{graphicx}        
\usepackage{multicol}        
\usepackage{multirow}
\usepackage[bottom]{footmisc}
\usepackage{amsmath}                      
\usepackage{amsfonts}                     
\usepackage{mathrsfs}                     
\usepackage{amssymb}
\usepackage{float}
\usepackage{graphicx} 
\usepackage{nicefrac}
\usepackage{subfig}
\usepackage[vlined,linesnumbered,boxruled]{algorithm2e}
\graphicspath{../graphics/}
\usepackage[svgnames]{xcolor}

\usepackage{rotating}
\usepackage{sidecap}

\usepackage[dvips,frame,rotate,import]{xy}                   

\DeclareMathOperator*{\argmin}{arg\,min}

\newcommand{\nonnegS}{\mbox{N-SSSP}\xspace}
\newcommand{\posS}{\mbox{P-SSSP}\xspace}
\newcommand{\nonnegBoundedS}{\mbox{BN-SSSP}\xspace}
\newcommand{\posBoundedS}{\mbox{BP-SSSP}\xspace}

\newcommand{\bigOh}[1]{{O}(#1)}

\newcommand{\bigTheta}[1]{{\Theta}(#1)}

\newcommand{\graph}{G}
\newcommand{\edgeSet}{E}
\newcommand{\vertexSet}{V}

\newcommand{\numNodes}{n}
\newcommand{\numEdges}{m}

\newcommand{\parent}{p}

\newcommand{\shortestPathTree}{S}

\newcommand{\guess}{D}
\newcommand{\actual}{d}
\newcommand{\key}{\hat{D}}

\newcommand{\edge}{\varepsilon}
\newcommand{\node}{v}
\newcommand{\startNode}{s}

\newcommand{\pathSequence}{P}

\newcommand{\edgeLength}[1]{\|#1\|}

\newcommand{\minEdgeLength}{\delta}
\newcommand{\maxEdgeLength}{C}

\newcommand{\heap}{H}

\newcommand{\wordLen}{w}

\newcommand{\smallThing}{\epsilon}

\newcommand{\updateValue}[2]{\mathtt{updateValue}(#1,#2)}
\newcommand{\extractMin}[1]{\mathtt{extractMin}(#1)}

\newcommand{\fullOrdering}{\mbox{FO}\xspace}
\newcommand{\partialOrdering}{\mbox{$\minEdgeLength$-PO}\xspace}
\newcommand{\alternativeOrdering}{\mbox{AO}\xspace}

\providecommand{\dontprintsemicolon}{}

\begin{document}

\title{On Solving Floating Point SSSP Using an Integer Priority Queue}
\titlerunning{On Solving Floating Point SSSP Using an Integer Priority Queue}  
%
\author{Michael Otte}
\authorrunning{Michael Otte} 
%
\tocauthor{Michael Otte}
\institute{Work done at University of Colorado at Boulder\\ Dept. Aerospace Engineering Sciences}

\maketitle              

\begin{abstract}
We address the single source shortest path planning problem (SSSP) in the case of floating point edge weights.
We show how any {\it integer} based Dijkstra solution that relies on a monotone {\it integer} priority queue to create a full ordering over path lengths in order to solve {\it integer} SSSP can be used as an oracle to solve {\it floating point} SSSP with positive edge weights (floating point \posS).  Floating point \posS is of particular interest to the robotics community.
This immediately yields a handful of faster runtimes for floating point \posS; for example, ${\bigOh{\numEdges + \numNodes\log \log \frac{\maxEdgeLength}{\minEdgeLength}}}$, where  $\maxEdgeLength$ is the largest weight and $\minEdgeLength$ is the minimum edge weight in the graph. It also ensures that many future advances for integer SSSP will be transferable to floating point \posS.
Our work relies on a result known to \cite{dinic.78} and \cite{tsitsiklis.95}; that, under the right conditions, SSSP can be solved using a partial ordering of nodes, despite the fact that full orderings are typically used in practice. Thus, priority queues that {\it do not} extract keys in a monotonically nondecreasing order can be used to solve SSSP under a special set of conditions that always hold for floating point \posS.
In particular, any node that has a shortest-path-length (key value) within $\minEdgeLength$ of the queue's minimum key can be extracted (instead of the node with the minimum key) from the priority queue and Dijkstra's algorithm will still run correctly.
Our contribution is to show how {\it any} monotonic {\it integer} priority queue can be transformed into a suitable $\minEdgeLength$-nonmonotonic {\it floating point} priority queue in order to produce the necessary partial ordering for floating point values. Monotonic integer keys for floating point values are obtained by dividing the floating point key values by $\minEdgeLength$ (as was done by \cite{dinic.78}) and then converting the result to an intege. The loss of precision this division causes is the reason the resulting floating point queue is   $\minEdgeLength$-nonmonotonic instead of monotonic, but does not break the correctness of Dijkstra's algorithm.
We also prove that the floating point version of {\it non-negative} SSSP (which allows $\minEdgeLength = 0$ in addition to $\minEdgeLength > 0$) cannot be solved without creating a full ordering of nodes (and thus requires a fully monotonic priority heap); and so our method cannot be extended to work on non-negative floating point SSSP, in general.

\end{abstract}

\section{Introduction} \label{sec:introduction}

Finding the shortest path through a graph $\graph = (\edgeSet, \vertexSet)$ defined by sets of nodes ${\vertexSet}$ and edges $\edgeSet$ is a classic problem.
%
%
The  variation known as the ``single source shortest-path planning problem'' (or SSSP) is concerned with finding all of the shortest-paths from a particular node ${\startNode \in \vertexSet}$ to all other nodes ${\node \in \vertexSet}$, where each edge ${(u,\node) = \edge \in \edgeSet}$ is associated with a length $\edgeLength{\edge}$. 
The first algorithm that solves SSSP was presented by Dijkstra in the 1950s using an algorithm that runs in time $\bigOh{\numEdges + \numNodes^2}$ for the case of nonnegative integer edge lengths \cite{Dijkstra}. 
Over the years, the SSSP problem has been studied extensively by computer scientists, who have primarily focus on the classic integer version of the problem. The robotics community, on the other hand, is more interested in floating point solutions to SSSP. 
In robotics, graphs are embedded in the environment\footnote{In many cases the graph used for robotics is embedded within the configuration space of the robot (the space created as the product space over the degrees of freedom of the system), of which the physical environment is  a lower dimensional projection.} and edge lengths represent  distances along physical trajectories that a robot traverses to move between different locations. Movement is subject to the practical geometric constraints of the environment (topology and obstacles) as well as the dynamics and kinomatics of the robotic system. Thus, finding suitable trajectories (and their edge lengths) often involves solving a set of differential equations or other two-point boundary value problem. In short, the most natural numerical representation for robotics SSSP problems is not integers.

\newcommand{\hspc}{@{\hspace{1em}}}
\definecolor{shadecolor}{RGB}{255,255,000}
\newcommand{\shadeThis}[1]{\colorbox{shadecolor}{#1}}
\newcommand{\citeTiny}[1]{\tiny{\cite{#1}}}
\newcommand{\citeFoot}[1]{\footnotesize{\cite{#1}}}
\newcommand{\rot}{\rotatebox{90}}

\begingroup
\def\arraystretch{1.5}
\begin{table}[t]
\caption{Best theoretical runtimes prior to our work (with references)} \label{table:before}

\makebox[\textwidth][c]{
\centering
\begin{tabular}{l \hspc r | l \hspc c l |  l \hspc l |}
& & \multicolumn{3}{c |}{Directed (SSSP)} & \multicolumn{2}{c |}{Undirected (USSSP)}\\
\hline

\multirow{9}{.2cm}{\rot{\normalsize{Integer}}} & 
          \multirow{6}{.2cm}{\rot{\small{worst case}}}&
                       $\bigOh{\numEdges + \numNodes \log \log \numNodes}\,\,$ &
                       &
                       \citeFoot{Thorup4} &
                       $\bigOh{\numEdges + \numNodes}$ & 
                       \citeFoot{Thorup}\\

        & 
                     &
                       $\bigOh{\numEdges + \numNodes \log \log C}\,\,$ &
                       &
                       \citeFoot{Thorup4} &
                       &
                       \\

        & 
                     &
                       $\bigOh{\numEdges + \numNodes(\log C \log \log C)^{1/3}}\,\,$ &
                       &
                       \citeFoot{Raman2} & 
                       &
                       \\

        &
                    &
                       ${\bigOh{\numNodes + \numEdges\log\log\numEdges}}$&
                       &
                       \citeFoot{Thorup3}&
                       &
                       \\

        &
                   &
                      ${\bigOh{\numNodes + \numEdges\log\log C}}$ &
                      &
                      \citeFoot{Raman2} & 
                      &
                      \\

        &
                   &
                      ${\bigOh{\numEdges + \numNodes(B+C/B)}}$ s.t. ${B < C+1}$ & 
                      &
                      \citeFoot{Cherkassky.etal.MP96} & 
                      &
                      \\

        &
                   &
                      ${\bigOh{\numEdges\Delta + \numNodes(\Delta + C/\Delta)}}$ s.t. $0 < \Delta < C$&
                      &
                      \citeFoot{Cherkassky.etal.MP96} &
                      &
                      \\

\cline{2-7}
               &
        \multirow{3}{.2cm}{\rot{\normalsize{expected}}} & 
                      $\bigOh{\numEdges + \numNodes \log \log \numNodes}\,\,$ &
                       &
                       \citeFoot{Thorup4}&
                      &
                      \\

        &
                    & 
                      $\bigOh{\numEdges + \numNodes \log \log C}\,\,$ &
                       &
                       \citeFoot{Thorup4} &
                      &
                      \\

        &
                    & 
                      ${\bigOh{\numEdges\sqrt{\log \log \numNodes }}}$&
                      &
                      \citeFoot{Thorup5}&
                      &
                      \\

\hline
\multirow{5}{.2cm}{\rot{\normalsize{Floating Point}}} &
         \multirow{3}{.2cm}{\rot{\small{worst case}}} &
                      ${\bigOh{\numEdges + \numNodes\log \log\numNodes}}$&
                      &
                      \citeFoot{Thorup4}&
                      $\bigOh{\numEdges + \numNodes}$&
                      \citeFoot{Thorup2}\\

               &
                    &
                     $\bigOh{\numEdges + \numNodes \frac{\log \frac{\maxEdgeLength}{\minEdgeLength}}{\log \log \frac{\maxEdgeLength}{\minEdgeLength}}}$ &
                     *&
                     \citeFoot{Cherkassky.etal}&
                     &
                     \\

               &
                    &
                      &
                      &
                      + \citeFoot{tsitsiklis.95}&
                     &
                     \\

\cline{2-7}
        &
           \multirow{3}{.2cm}{\rot{\small{expected}}} &
                      ${\bigOh{\numEdges + \numNodes\log \log\numNodes}}$&
                      &
                      \citeFoot{Thorup4}&
                      &
                      \\

        & 
                 &
                      $\bigOh{\numEdges + \numNodes(\log \frac{\maxEdgeLength}{\minEdgeLength})^{1/3+\epsilon}}$ &
                      *&
                      \citeFoot{Cherkassky.etal}&
                      &
                      \\

        &
                    & 
                      &
                      &
                      + \citeFoot{tsitsiklis.95}&
                      &
                      \\

        &
                    & 
                      ${\bigOh{\numEdges\sqrt{\log \log \numNodes }}}$&
                      &
                      \citeFoot{Thorup5}&
                      &
                      \\

\hline
\end{tabular}
}

\makebox[\textwidth][c]{
\centering
\begin{tabular}{l}
Results with (*) assume strictly {\bf positive edges}, (\posS), ${0 < \minEdgeLength \leq \edgeLength{\edge} \leq \maxEdgeLength < \infty}$,  ${\forall \edge \in \edgeSet}$.\\ Other results hold for both positive and  {\bf nonnegative} edges (\nonnegS), ${0 \leq \edgeLength{\edge} \leq \maxEdgeLength < \infty}$. \\
\end{tabular}
}

\end{table}
\endgroup

The main contribution of this paper is to show how a large class of solutions for integer SSSP can easily be modified to solve floating point SSSP with strictly positive edge weights (floating point \posS). Our work is applicable to both theoretical performance bounds, as well as the practical implementation of SSSP algorithms. We note that our method breaks if edges are allowed to have length $0$. This is usually not a concern for robotics applications, where an edge of length $0$ implies that the two nodes it connects represent the same location in the environment\footnote{Or configuration space.} --- without loss of generality\footnote{At least in the robotics domain, where the graph is just a tool for reasoning about environmental connectivity with respect to the movement constraints of the robot.} these can be combined into a single node and the $0$ length edge removed from the graph.

\begingroup
\def\arraystretch{1.5}
\begin{table}[h!]
\caption{Best theoretical runtimes after our work (with references)} \label{table:after}

\makebox[\textwidth][c]{
\centering
\begin{tabular}{l \hspc r | l \hspc c l |  l \hspc l |}
& & \multicolumn{3}{c |}{Directed (SSSP)} & \multicolumn{2}{c |}{Undirected (USSSP)}\\
\hline

\multirow{9}{.2cm}{\rot{\normalsize{Integer}}}  & 
          \multirow{6}{.2cm}{\rot{\normalsize{worst case}}} & 
                       $\bigOh{\numEdges + \numNodes \log \log \numNodes}\,\,$ &
                       &
                       \citeFoot{Thorup4} &
                       $\bigOh{\numEdges + \numNodes}$ & 
                       \citeFoot{Thorup}\\

        & 
                     &
                       $\bigOh{\numEdges + \numNodes \log \log C}\,\,$ &
                       &
                       \citeFoot{Thorup4} &
                       &
                       \\

        & 
                     &
                       $\bigOh{\numEdges + \numNodes(\log \maxEdgeLength \log \log \maxEdgeLength)^{1/3}}\,\,$ &
                       &
                       \citeFoot{Raman2} & 
                       &
                       \\

        &
                    &
                       ${\bigOh{\numNodes + \numEdges\log\log\numEdges}}$&
                       &
                       \citeFoot{Thorup3}&
                       &
                       \\

        &
                   &
                      ${\bigOh{\numNodes + \numEdges\log\log \maxEdgeLength}}$ &
                      &
                      \citeFoot{Raman2} & 
                      &
                      \\

        &
                   &
                      ${\bigOh{\numEdges + \numNodes(B+\maxEdgeLength/B)}}$ s.t. ${B < \maxEdgeLength+1}$ & 
                      &
                      \citeFoot{Cherkassky.etal.MP96} & 
                      &
                      \\

        &
                   &
                      ${\bigOh{\numEdges\Delta + \numNodes(\Delta + \maxEdgeLength/\Delta)}}$ s.t. $0 < \Delta < C$&
                      &
                      \citeFoot{Cherkassky.etal.MP96} &
                      &
                      \\

\cline{2-7}
        &
           \multirow{3}{.2cm}{\rot{\normalsize{expected}}} & 
                      $\bigOh{\numEdges + \numNodes \log \log \numNodes}\,\,$ &
                       &
                       \citeFoot{Thorup4}&
                      &
                      \\

        &
                    & 
                      $\bigOh{\numEdges + \numNodes \log \log C}\,\,$ &
                       &
                       \citeFoot{Thorup4} &
                      &
                      \\

        &
                    & 
                      ${\bigOh{\numEdges\sqrt{\log \log \numNodes }}}$&
                      &
                      \citeFoot{Thorup5}&
                      &
                      \\

\hline
\multirow{10}{.2cm}{\rot{\normalsize{Floating Point}}}  & 
          \multirow{7}{.2cm}{\rot{\normalsize{worst case}}} & 
                      ${\bigOh{\numEdges + \numNodes\log \log\numNodes}}$&
                      &
                      \citeFoot{Thorup4}&
                      $\bigOh{\numEdges + \numNodes}$&
                      \citeFoot{Thorup2}\\

        &
                    &
                      \shadeThis{${\bigOh{\numEdges + \numNodes\log \log \frac{\maxEdgeLength}{\minEdgeLength}}}$} & 
                      *&
                      \footnotesize{new result} & 
                      &
                      \\

       & 
                     &
                       \shadeThis{$\bigOh{\numEdges + \numNodes(\log \frac{\maxEdgeLength}{\minEdgeLength} \log \log \frac{\maxEdgeLength}{\minEdgeLength})^{1/3}}$} &
                       *&
                       \footnotesize{new result} & 
                       &
                       \\

        &
                    &
                       \shadeThis{${\bigOh{\numNodes + \numEdges\log\log\numEdges}}$}&
                       *&
                       \footnotesize{new result} & 
                       &
                       \\

        &
                   &
                      \shadeThis{${\bigOh{\numNodes + \numEdges\log\log \frac{\maxEdgeLength}{\minEdgeLength}}}$}&
                      *&
                      \footnotesize{new result} & 
                      &
                      \\

        &
                   &
                      \shadeThis{${\bigOh{\numEdges + \numNodes(B+\frac{\maxEdgeLength}{\minEdgeLength B})}}$ s.t. ${B < \maxEdgeLength+1}$}& 
                      *&
                      \footnotesize{new result} & 
                      &
                      \\

        &
                   &
                      \shadeThis{${\bigOh{\numEdges\Delta + \numNodes(\Delta + \frac{\maxEdgeLength}{\minEdgeLength \Delta})}}$ s.t. $0 < \Delta < C$}&
                      *&
                      \footnotesize{new result} & 
                      &
                      \\

\cline{2-7}
        &
                     \multirow{3}{.2cm}{\rot{\small{expected}}} &
                      ${\bigOh{\numEdges + \numNodes\log \log\numNodes}}$&
                      &
                      \citeFoot{Thorup4}&
                      &
                      \\

        &
                    &
                      \shadeThis{${\bigOh{\numEdges + \numNodes\log \log \frac{\maxEdgeLength}{\minEdgeLength}}}$} & 
                      *&
                      \footnotesize{new result} & 
                      &
                      \\

        &
                    & 
                      ${\bigOh{\numEdges\sqrt{\log \log \numNodes }}}$&
                      &
                      \citeFoot{Thorup5}&
                      &
                      \\

\hline
\end{tabular}
}

Results with (*) assume strictly {\bf positive edges}, (\posS), ${0 < \minEdgeLength \leq \edgeLength{\edge} \leq \maxEdgeLength < \infty}$,  ${\forall \edge \in \edgeSet}$. Other results hold for both positive and  {\bf nonnegative} edges (\nonnegS), ${0 \leq \edgeLength{\edge} \leq \maxEdgeLength < \infty}$. 
Yellow indicates new results enabled by our work.

\end{table}
\endgroup

The vast majority of previous work on SSSP has involved graphs with integer edge weights (related work is surveyed in Section~\ref{sec:relatedwork}). 
Sophisticated priority queue data structures have reduced the runtime of integer SSSP algorithms to ${\bigOh{\numEdges + \numNodes\log\numNodes}}$ for the restrictive {\it comparison} computational model and even further for the {\it RAM} computational model that most ``real'' computers use e.g., ${\bigOh{\numEdges + \numNodes\log \log\numNodes}}$. See Table~\ref{table:before} for additional results.

The discrepancy between the state-of-the-art for integer vs.\ floating point SSSP can be seen in Table~\ref{table:before}.
That said, notable floating point results include \cite{Thorup2} and \cite{Thorup4}. \cite{Thorup2} proves that, in the case of {\it undirected} graphs (i.e., USSSP), an algorithm that solves integer USSSP can be used as an oracle to solve floating point USSSP. \cite{Thorup4} leverages the fact that, assuming the IEEE standard for floating point numbers and integers is used, value orderings are preserved if we interpret the bit representations of floating point numbers directly as integers. The latter trick enables many (but not all) results from integer SSSP to be extended to floating point SSSP, assuming IEEE standards are followed. In particular, it enables the direct application of the ${\bigOh{\numEdges + \numNodes\log \log\numNodes}}$ integer SSSP algorithm presented in \cite{Thorup4} to be used to solve floating point SSSP.

From a high-level point-of-view, our work extends the applicability of integer-to-floating-point oracles to a greater subset of SSSP. In particular, it enables any solution for integer SSSP that constructs a full ordering (i.e., by extracting values from an integer priority queue in a monotonically nondecreasing sequence) to solve floating point \posS. The floating point solution happens via a partial ordering on path lengths that can be found by extracting floating point values in a special sequence that may {\it not} be monotonically nondecreasing. 

The trick that we use is fundamentally different than interpreting the bits of floating point numbers directly as integers, and thus applies to a different subset of SSSP algorithms than previous work. 
Many (and arguably most) integer SSSP solutions construct a full ordering based on path-length.
Leveraging this previous work, we immediately get a variety of faster theoretical runtimes for floating point \posS.
For example, combining our work with \cite{Thorup4} yields an algorithm for floating point \posS that runs in time ${\bigOh{\numEdges + \numNodes\log \log \frac{\maxEdgeLength}{\minEdgeLength}}}$, where  $\minEdgeLength$ is the minimum edge length in the graph. This is particularly of note given the comment in \cite{Thorup4} that ``for floating point numbers we [they] do not get bounds in terms of the maximal weight,'' i.e., $\maxEdgeLength$.
Other new results enabled by our work appear highlighted in Table \ref{table:after}.
%
Our works also guarantees  that any new faster results for integer SSSP that come from the discovery of a new monotonic integer priority queue will have corresponding counterparts for floating point \posS.
%

\subsection{High Level Description}

Our method relies on a result known to \cite{dinic.78} and \cite{tsitsiklis.95}; that, under the right conditions, SSSP can be solved using a partial ordering of nodes, despite the fact that full orderings are typically used in practice.
%
%
In particular, any node that has a shortest-path-length (key value) within $\minEdgeLength$ of the queue's minimum key can be extracted (instead of the node with the minimum key) from the priority queue and Dijkstra's algorithm will still run correctly.
We show how an integer priority queue that extracts values in a monotonically nondeceasing order can easily be converted into a floating point priority queue that is suitable for solving floating point \posS. This is acomplished  by calculating  {\it integer keys} for {\it floating point values} by dividing the latter by $\minEdgeLength$ and then converting with truncation to an integer representation. The resulting floating point priority is no longer monotonic, but rather $\minEdgeLength$-nonmonotonic (see Section~\ref{sec:monotone} for a discussion on the differences between $\minEdgeLength$-nonmonotonic and nonmonotonic). While \cite{dinic.78} uses a similar idea to improve integer bucket-based priority queues for solving integer SSSP, we are the first to employ this mechanism to leverage integer priority queues for solving floating point SSSP.

Algorithmically this means that we store two keys with each node instead of one. The actual floating point key $\guess(\node)$ that reflects the current best guess of node $\node$'s shortest path length, as well as a companion integer key $\key(\node) = \lfloor\frac{\guess(\node)}{\minEdgeLength} \rfloor$.  We maintain the position of node $\node$ within the integer queue using $\key(\node)$ instead of $\guess(\node)$. Whenever $\guess(\node)$  is decreased we recalculate $\key(\node)$ from the new $\guess(\node)$ and then update $\node$'s position in the integer queue using $\key(\node)$. The differences between standard Dijkstra's algorithm and a version using our modification to solve floating point \posS can be seen by comparing Algorithms~\ref{alg0}~and~\ref{alg1}.

Although there is loss of precision within the integer priority queue (due to conversion from floating point to integer with truncation) any resulting node reordering is upper-bounded by $\minEdgeLength$ with respect to $\guess(\node)$. Thus, as long as ${\minEdgeLength > 0}$, any reordering of nodes that results due to this precision loss is still sufficient to solve floating point SSSP with Dijkstra's algorithm (i.e., the resulting $\minEdgeLength$-nonmonotonic extraction order from the resulting floating point queue creates a partial ordering sufficient for Dijkstra's algorithm to solve SSSP).
The rest of this paper is dedicated to formalizing these results and discussing them in greater detail.

\section{Preliminaries} \label{sec:prelim}


A graph $\graph$ (either directed or underacted) is defined by its edge set $\edgeSet$ and vertex set $\vertexSet$. We assume that both ${\numEdges = |\edgeSet|}$ and ${\numNodes = |\vertexSet|}$ are finite. 
Each edge $\edge_{ij} = (\node_i,\node_j)$ between two vertices $\node_i$ and $\node_j$ (or from $\node_i$ to $\node_j$ if the edge is directed) is assumed to have a predefined length (or edge-length or cost) $\edgeLength{\edge_{ij}} = \edgeLength{(\node_i,\node_j)}$ such that ${0 \leq \edgeLength{\edge_{ij}} \leq \infty}$ for all $\edge_{ij}\in \edgeSet$. 


A path ${\pathSequence(\node_i,\node_j)}$ is an ordered sequence of edges ${\edge_1,\ldots, \edge_\ell}$ such that ${\edge_1 = (\node_i,\node_1)}$ and ${\edge_k = (\node_{k-1}, \node_{k})}$ for all ${k \in \{2, \ldots, \ell-1\}}$ and ${\edge_\ell = (\node_{\ell-1},\node_j)}$ and where ${\{\edge_k \in \pathSequence(\node_i,\node_j)\} \subset \vertexSet}$. The shortest path $\pathSequence^*(\node_i,\node_j)$ is the shortest possible path from $\node_i$ to $\node_j$. Formally,
$${
\pathSequence^*(\node_i,\node_j) \equiv \argmin_{\pathSequence(\node_i,\node_j)} \sum_{\edge \in \pathSequence(\node_i,\node_j)}  \edgeLength{\edge}
}$$
We are interested in paths from a particular ``start-node'' $\startNode$ to nodes $\node$, and define $\actual(\node)$ to be the length of the shortest possible path from $\startNode$ to $\node$.
$${
\actual(\node) \equiv \min_{\pathSequence(\startNode,\node)} \sum_{\edge \in \pathSequence(\startNode,\node)}  \edgeLength{\edge}
}$$

\begin{figure}[t!]
    \begin{algorithm}[H] \label{alg0}
      \dontprintsemicolon
      \caption{$\mathtt{Dijkstra}(\graph, \startNode)$  standard implementation} 
      \KwIn{A graph $\graph = (\edgeSet, \vertexSet)$ of node set ${\vertexSet}$ and edge set $\edgeSet$, a start node $\startNode \in \vertexSet$, and a priority queue $\heap$.}
      \KwOut{Shortest path lengths $\actual(\node)$ and parent pointers $\parent(\node)$ with respect to the shortest path-tree $\shortestPathTree$ for all $\node \in \vertexSet$.} 

        $\heap = \emptyset$ \;                          \label{alg0:a}
        \For{$\mathbf{all}$ $\node \in \vertexSet$}
        {                                                \label{alg0:b}
          $\guess(\node) = \infty$ \;                    
          $\parent(\node) = \mathrm{NULL}$ ;   \hspace{3.8cm} /* $\shortestPathTree = \emptyset$ */ \;                    \label{alg0:c}
        }
        $\guess(\startNode) = 0$ \;                       \label{alg0:d}
        $\node = \startNode$ \;                            \label{alg0:e}
        \While{$\node \neq \mathrm{NULL}$} 
        {                                                  \label{alg0:f}
          $\actual(\node) = \guess(\node)$ ;  \hspace{4cm} /* $\shortestPathTree = \shortestPathTree \cup \{\node\}$ */ \;    \label{alg0:g}
          \For{$\mathbf{all}$ $u$ $\mathrm{s.t.}$ $(v,u) \in \edgeSet$}
          {                                                   \label{alg0:h}
            \If{$\actual(\node) + \edgeLength{(v,u)} < \guess(u)$}
            {                                                   \label{alg0:i}
              $\guess(u) = \actual(\node) + \edgeLength{(v,u)}$ \;       \label{alg0:j}
              $\parent(u) = \node$ \; 
              $\updateValue{\heap}{u,\guess(u)}$ \;                \label{alg0:k}
            }
          }
          $\node = \extractMin{\heap}$ \;                     \label{alg0:l}
        }         
    \end{algorithm}
\end{figure}

\begin{figure}[t!]

    \begin{algorithm}[H] \label{alg1}
      \dontprintsemicolon
      \caption{$\mathtt{Dijkstra}(\graph, \startNode)$  our modification that solves {\it floating point} \posS using an {\it integer} priority queue.} 
      \KwIn{A graph $\graph = (\edgeSet, \vertexSet)$ of node set ${\vertexSet}$ and edge set $\edgeSet$, a start node $\startNode \in \vertexSet$, and a monotone {\it integer} priority queue $\heap$.}
      \KwOut{Shortest path lengths $\actual(\node)$ and parent pointers $\parent(\node)$ with respect to the shortest path-tree $\shortestPathTree$ for all $\node \in \vertexSet$.} 

        $\heap = \emptyset$ \;                        
        \For{$\mathbf{all}$ $\node \in \vertexSet$}
        {                                         
          $\guess(\node) = \infty$ \;                    
          $\parent(\node) = \mathrm{NULL}$ ;   \hspace{3.8cm} /* $\shortestPathTree = \emptyset$ */ \;                    \label{alg0:c}
        }
        $\minEdgeLength = \infty $\;                 \label{alg1:1}
        \For{$\mathbf{all}$ $(v,u) \in \edgeSet$}
        {                                               
          \If{$\edgeLength{(v,u)} < \minEdgeLength$}
          {
            $\minEdgeLength = \edgeLength{(v,u)}$\;     \label{alg1:2}
          }
        }
        $\guess(\startNode) = 0$ \;                  
        $\node = \startNode$ \;                         
        \While{$\node \neq \mathrm{NULL}$} 
        {                                                
          $\actual(\node) = \guess(\node)$ ;  \hspace{4cm} /* $\shortestPathTree = \shortestPathTree \cup \{\node\}$ */ \;    
          \For{$\mathbf{all}$ $u$ $\mathrm{s.t.}$ $(v,u) \in \edgeSet$}
          {                                                
            \If{$\actual(\node) + \edgeLength{(v,u)} < \guess(u)$}
            {                                                
              $\guess(u) = \actual(\node) + \edgeLength{(v,u)}$ \;    
              $\key(\node) = \lfloor\frac{\guess(\node)}{\minEdgeLength} \rfloor$  \hspace{2.7cm} /* ${\lfloor  \cdot \rfloor}$ converts to integer with truncation */ \;  \label{alg1:3}
              $\parent(u) = \node$ \; 
              $\updateValue{\heap}{u,\key(\node)}$ \;               
            }
          }
          $\node = \extractMin{\heap}$ \;                  
        }         
    \end{algorithm}
\end{figure}

Dijkstra's algorithm works by incrementally building a ``shortest-path-tree'' $\shortestPathTree$ outward from $\startNode$, one node at a time. Dijkstra's algorithm appears in Algorithm~\ref{alg0}. Each node that is not yet part of the growing $\shortestPathTree$ refines a ``best-guess'' $\guess(\node)$ of its actual shortest-path-length $\actual(\node)$, with the restriction that ${\actual(\node) \leq \guess(\node)}$.
Dijkstra's algorithm guarantees/requires that ${\actual(\node) = \guess(\node)}$ for the node in ${\vertexSet \setminus \shortestPathTree}$ with minimum $\guess(\node)$. In modern versions of the algorithm, a min-priority-heap $\heap$ (or related data structure) is used to keep track of $\guess(\node)$ values.

The min-priority-heap $\heap$ is initialized to empty, 
best-guesses $\guess(\node)$ are initialized to $\infty$, 
parent pointers $\parent(\node)$ with respect to the shortest path tree $\shortestPathTree$ are initialized to $\mathrm{NULL}$,
and the start node $\startNode$ is given an actual distance of $0$ from itself, lines \ref{alg0:a}-\ref{alg0:d}, respectively.

Each iteration involves \mbox{``processing''} the node ${\node \in \vertexSet \setminus \shortestPathTree}$ with minimum $\guess(\node)$ lines \ref{alg0:g}-\ref{alg0:k} (where $\node$ is removed, for all practical purposes, from future consideration during the algorithm's execution). Such a node $\node$ is extracted from the heap on line~\ref{alg0:l} (in the first iteration we know to use $\startNode$, line~\ref{alg0:e}). Next, each neighbor $u$ of $\node$ checks if ${\actual(\node) + \edgeLength{(\node,u)} < \guess(u)}$, that is, if the distance from $\startNode$ through the shortest-path-tree to $\node$  plus the distance from $\node$ to $u$ through edge ${(\node, u) \in \edgeSet}$  is less than $u$'s current best-guess. If so, then $u$ updates its best-guess and parent pointer to reflect the better path via $\node$,  lines \ref{alg0:i}-\ref{alg0:k}. All neighbors $u$ of $\node$ perform the update ${\guess(u) = \min(\guess(u), \actual(\node) + \edgeLength{(\node, u)})}$. The heap is adjusted to account for changing $\guess(u)$ on line \ref{alg0:k}.

Our modification that enables any integer based solution (i.e., a monotone integer priority queue) to be used for a floating point problem appears in Algorithm~\ref{alg1}. It relies on knowledge of $\minEdgeLength$, the length of shortest edge in $\edgeSet$. 
\begin{equation} \label{eq:min}
\minEdgeLength = \min_{\edge \in \edgeSet} \edgeLength{\edge}
\end{equation}
If $\minEdgeLength$ is not known a priori then it can be obtained by scanning all edges in the graph in time $\bigOh{\numEdges}$, lines~\ref{alg1:1}-\ref{alg1:2}.
We store two keys with each node instead of one. The actual floating point key $\guess(\node)$ that reflects the current best guess of node $\node$'s shortest path length, as well as a companion integer key $\key(\node) = \lfloor\frac{\guess(\node)}{\minEdgeLength} \rfloor$, line~\ref{alg1:3}, that is used to update the position of $\node$ within the integer priority queue.

A number of the previous (integer-based) algorithms that are compatible with our method have runtime dependent on the maximum weight $\maxEdgeLength$ 
%
%
%
and $\wordLen$, the word size of the computer being used (currently 32 or 64 in most computers). 
The effects of $\wordLen$ and $\maxEdgeLength$ on runtime are buried in the heap operations, and dependent on the particular algorithm being used. If a value of $\maxEdgeLength$ is explicitly needed by a particular priority queue then it can be found by scanning edges at the beginning of the algorithm, similarly to $\minEdgeLength$.

\subsection{Problem Variants} \label{sec:problem}

\subsubsection{SSSP} 

The single source shortest path planning problem is defined:

{\it 
Given $\graph = (\vertexSet, \edgeSet)$ and a particular node ${\startNode  \in \vertexSet}$, then for all ${\node \in \vertexSet}$, find the shortest path $\pathSequence^*(\startNode, \node)$. }\\

The problems is considered solved once we have produced a data structure containing both:
\begin{enumerate} 
\item The shortest-path lengths $\actual(\node)$ for all $\node \in \vertexSet$ from $s$.
\item The shortest path tree that can be used to extract the shortest path from $s$ to any $\node$ (at least for any $\node$ such that $\actual(\node) < \infty)$.
\end{enumerate}
The latter can be accomplished by storing the parent of each node with respect to $\shortestPathTree$, allowing each shortest path to be extracted by following back pointers in the fashion of gradient descent from $v$ to $\startNode$ and then reversing the result.

The reverse (i.e., sink) search that involves finding all paths to $\startNode$ (instead of from $\startNode$) can be solved using basically the same algorithm except that the rolls played by in- and out- neighbors are swapped and the extracted path is not reversed. 

For the sake of clarity and brevity in the rest of our presentation, we now formally define abbreviations for a number of variations on SSSP. Variations are related to the range of allowable edge weights (e.g., whether or not zero-length edges allowed), and if edges are directed or undirected.

\subsubsection{\nonnegS}

``Nonnegative SSSP'' refers to the special case of SSSP such that edges have nonnegative weights, $0 \leq \edgeLength{\edge}$ for all $\edge \in \edgeSet$.

\subsubsection{\posS}

``Positive SSSP'' refers to the special case of SSSP such that edges have positive weights, $0 < \minEdgeLength \leq \edgeLength{\edge}$ for all $\edge \in \edgeSet$.

\subsubsection{\nonnegBoundedS}

``Bounded Nonnegative SSSP'' refers to the special case of \nonnegS such that edge weights are bounded above, ${0 \leq \edgeLength{\edge} \leq \maxEdgeLength < \infty}$ for all $\edge \in \edgeSet$.

\subsubsection{\posBoundedS}

``Bounded Positive SSSP'' refers to the special case of \posS such that edge weights are bounded above, ${0 < \edgeLength{\edge} \leq \maxEdgeLength < \infty}$ for all $\edge \in \edgeSet$.

\subsubsection{USSSP} 

The single source shortest path planning problem for undirected graphs is defined:

{\it 
Given $\graph = (\vertexSet, \edgeSet)$ such that for all $(u,v) \in \edgeSet$ there exists $(v,u) \in \edgeSet$ such that $\edgeLength{(u,v)} = \edgeLength{(v,u)}$, and a particular node ${\startNode  \in \vertexSet}$, then for all ${\node \in \vertexSet}$, find the shortest path $\pathSequence^*(\startNode, \node)$. }

The prefixes N, P, BN, and BP can also be used with USSSP, e.g., P-USSSP is ``Positive USSSP.''

\subsubsection{Floating point vs.\ integer}

We explicitly differentiate between variants of problems that use either integer or floating point edge weights. For example, integer SSSP is the subset of SSSP that uses integer weights; floating point \posBoundedS is the subset of \posBoundedS that uses floating point weights; etc.

\subsection{Orderings} \label{sec:orderings}

Dijkstra's original algorithm is provably correct (see \cite{Dijkstra}), based on guarantees that the next node $v$ processed at any step has the following properties:
\begin{enumerate}
\item $\node  \in \vertexSet \setminus \shortestPathTree$.
\item Either $\node = \startNode$ or $\node$ has some neighbor $u$ such that $u \in \shortestPathTree$.
\item $\guess(\node) \leq \guess(v')$, for all nodes $v' \in \vertexSet \setminus \shortestPathTree$.
\end{enumerate}
As has often been remarked, it works by finding a full ordering on $\actual(\node)$ for all ${v \in \vertexSet}$. The priority heap data structure enforces (3) and determines this ordering. 

However, a partial ordering based on $\minEdgeLength$ is sufficient to solve \posS; this is formally proven in Section~\ref{sec:po} and has previously been shown by \cite{tsitsiklis.95} and, according to \cite{Thorup}, even earlier in Russian by \cite{dinic.78}. We make extensive use of the latter partial ordering in the current paper.

Let ${v_1, \ldots, v_n}$ denote the order in which Dijkstra's algorithm processes nodes.

\subsubsection{\fullOrdering (Full Ordering)}

A full ordering of nodes based on path-length level-sets is a sequence ${v_1, \ldots, v_n}$ such that: 
$$
i < j \Rightarrow \actual(\node_i) \leq \actual(\node_j).
$$
in other words, the fact that the $j$-th node is extracted after the $i$-th node implies that the shortest path from the $j$-th node is at least as long as the one from the $i$-th node.

\subsubsection{\partialOrdering (Partial Ordering Based on $\minEdgeLength$)}

A partial ordering of nodes based on $\minEdgeLength$ values is a sequence ${v_1, \ldots, v_n}$ such that:
$$
i < j \Rightarrow \actual(\node_i) < \actual(\node_j) + \minEdgeLength
$$
where $\minEdgeLength$ is the minimum edge weight in the graph.
Note that \partialOrdering is similar to \fullOrdering except that nodes $\node$ and $u$ can be swapped if $\actual(\node)$ and  $\actual(u)$ are within $\minEdgeLength$.
As we will show shortly, \partialOrdering is sufficient for solving \posS.

\subsection{Monotonicity} \label{sec:monotone}

A {\it monotone priority queue} is, unfortunately, given a variety of definitions in the literature. According to \cite{Raman} ``A monotone priority queue has the property that the value of the minimum key stored in the priority queue is a don-decreasing function of time.'' 
\cite{Cherkassky.etal} alternatively state that ``In monotone priority queues the extracted keys form a monotone, nondecreasing sequence.''

While all heaps that meet the second definition necessarily meet the first; the reverse is not true, in general.
The difference between these two definitions is particularly important for our work. 
The queues we use to create \partialOrdering have the property that the minimum value {\it stored} in the priority queue is monotonically nondecreasing but do not, in general, {\it extract} nodes from the priority queue in a monotonically nondecreasing order. Moreover, they do not even guarantee that the node with the minimum value is extracted; but instead guarantee only that the node extracted from the heap has a key that is within $\minEdgeLength$ of the minimum.
We will use the term ``$\minEdgeLength$-nonmonotonic'' to describe the floating point heap we use/require, in order to highlight this important property.

We note that the class of $\minEdgeLength$-nonmonotonic priority queues are a {\it relaxed} version of ``monotone priority queues'' (where we use quotation marks to denote that either the first or second definition may be used). Thus our special $\minEdgeLength$-nonmonotonic priority queues can be assumed to require no more runtime or space than ``monotone priority queues''  and may even require less.
The special (useful) case of $\minEdgeLength$-nonmonotonicity should not be confused with standard nonmonotonicity.  A nonmonotonic priority queue always extracts the minimum value and makes no assumptions on the order, partial or otherwise, that key values are added to and/or extracted from the queue (since less assumptions can be made on the values stored in the queue, ``nonmonotonic priority queues'' require more runtime and space than ``monotone priority queues''; this is the opposite of our special case $\minEdgeLength$-nonmonotonicity).

\subsection{Computational Model}

We assume the existence of an algorithmic solution to integer \nonnegS that creates a \fullOrdering using a monotone integer priority queue. As such, we inherit whatever computational model is assumed by this underlying integer-based algorithm. 
In general, basic arithmetic and related operations on $\wordLen$-length words are assumed to require $\bigOh{1}$ time; where ${\wordLen = \bigOh{\log \maxEdgeLength}}$ so that extra computations cannot be ``hidden'' by performing them in parallel on unnecessarily ample numerical representations.
Previous  methods with runtimes dependent on $\maxEdgeLength$ assume that ${\wordLen \geq \log \maxEdgeLength}$ (i.e., so that it is possible to efficiently represent integers and floating point numbers in memory), and thus ${\wordLen = \bigTheta{\log \maxEdgeLength}}$.

We additionally assume that the computational model supports floating point numbers and that floating-point-to-integer conversion (with truncation) is supported in $\bigOh{1}$ time. Similar to other floating point algorithms, we also assume cumulative errors due to fixed floating point precision are tolerated within the definition of  ``correct solution.'' This is commonly handled by representing floating point numbers at twice the normal precision during mathematical operations within the CPU and then rounding to the nearest represented floating point number for RAM storage.

\section{Sufficiency of \partialOrdering Partial Ordering for Solving \posS} \label{sec:po}

We now prove that Dijkstra's algorithm will correctly solve (integer and floating point) \posS if nodes are processed according to \partialOrdering, the partial ordering based on $\minEdgeLength$ as defined in Section~\ref{sec:orderings}. We note that \cite{tsitsiklis.95} presents an alternative proof of the same idea.

Recall that ``processing'' $\node_i$ is the act of removing $\node_i$ from the heap and updating its neighbors ${\node \in \{\node_j \,\, | \,\, (\node_i,\node_j) \in \edgeSet \}}$ with respect to any path-length decreases that can be achieved via edges from $\node_i$.

\begin{lemma} \label{lem:po}
Assuming \partialOrdering exists, and nodes are processed according to \partialOrdering with Dijkstra's algorithm, and ${\actual(\node_i) = \guess(\node_i)}$ when $\node_i$ is processed for ${i = 1, \ldots, j-1}$, then ${\actual(\node_j) = \guess(\node_j)}$.
\end{lemma}
\begin{proof}
(by contradiction) Assume instead $\actual(\node_j) < \guess(\node_j)$.
Therefore, there must exist some $\node_k$ and edge $\edge_{kj}$ such that ${\actual(\node_k) + \edgeLength{\edge_{kj}} = \actual(\node_j) < \guess(\node_j)}$. This has two ramifications, first:
\begin{equation} \label{eq:a}
{\actual(\node_k) + \minEdgeLength \leq \actual(\node_j)}
\end{equation}
because $\minEdgeLength \leq \edgeLength{\edge_{kj}}$ by definition; and second:
$$
j<k
$$
(otherwise, if $k<j$, then we would have already set ${\guess(\node_j) = \actual(\node_j) = \actual(\node_k) + \edgeLength{\edge_{kj}}}$).
But by the definition of \partialOrdering the fact that $j<k$ implies: 
\begin{equation} \label{eq:b}
\actual(\node_k) + \minEdgeLength > \actual(\node_j)
\end{equation} 
yielding the necessary contradiction. 
\qed
\end{proof}

Now we are able to prove the correctness of Dijkstra's algorithm using \partialOrdering by showing that all nodes $\node$ are processed from the \partialOrdering heap only if ${\actual(\node) = \guess(\node)}$, i.e., the shortest path from $\node$ to $\startNode$ has already been computed.  

\begin{theorem} \label{th:correct}
Assuming a \partialOrdering exists, and nodes are processed according to \partialOrdering, then Dijkstra's algorithm correctly sets ${\actual(\node) = \guess(\node)}$ for all ${\node \in \vertexSet}$. 
\end{theorem}
\begin{proof}
The proof is by induction on $j$. The inductive step relies on Lemma~\ref{lem:po} for ${j = n, \ldots, 2}$. The base case for ${j = 1}$ is given by the fact that $\node_1 = \startNode$ and $\actual(\startNode) = \guess(\startNode) = 0$ before the algorithm starts.
\qed

\end{proof}

\section{Any Heap that Constructs an FO for Integer \nonnegS is an Oracle for \partialOrdering for Floating Point \posS } \label{sec:heaps_orical}

We now show how any monotone integer heap $\heap$ that produces a FO when given integer weights for integer \posS can be used as an oracle to create a $\minEdgeLength$-nonmonotonic floating point heap that produces a \partialOrdering for floating point \posS. Recall that (for floating point SSSP) the floating point key for a node $\node$ is given by $\guess(\node)$. 
We define a companion ``{\it integer} key'' $\key(\node) = \lfloor\frac{\guess(\node)}{\minEdgeLength} \rfloor$, and maintain the position  $\node$ within $\heap$ using $\key(\node)$ instead of $\guess(\node)$. In other words, whenever $\guess(\node)$  is decreased we also recalculate  $\key(\node)$ and then update $\node$ in $\heap$ using $\key(\node)$. 

By definition of the conversion with truncation operator ${\lfloor \cdot\rfloor}$  and the fact that ${\minEdgeLength > 0}$ we are guaranteed that ${\guess(\node) - \minEdgeLength\key(\node) < \minEdgeLength}$ for all nodes at any point during Dijkstra's execution. Assuming that ${i < j}$, then at the instant $\node_j$ is processed there are two cases:
$${\key(\node_j) - \key(\node_i) = 0}$$
and
$${\key(\node_j) - \key(\node_i) \geq 1}.$$
In the first case, ${\key(\node_j) - \key(\node_i) = 0}$ implies that ${|\guess(\node_i) - \guess(\node_j)| < \minEdgeLength}$ and so either ${\guess(\node_i) < \guess(\node_j) + \minEdgeLength}$, in the case that $\guess(\node_i) > \guess(\node_j)$, or ${\guess(\node_i) \leq \guess(\node_j)}$ otherwise --- and so ${\guess(\node_i) < \guess(\node_j) + \minEdgeLength}$ trivially.
In the second case ${\key(\node_j) - \key(\node_i) \geq 1}$ implies that ${\guess(\node_i) < \guess(\node_j)}$ and so $\guess(\node_i) < \guess(\node_j) + \minEdgeLength$ trivially. 
This discussion yields the following theorem:

\begin{theorem} \label{th:po}
Assuming a heap $\heap$ exists that creates a FO when used by Dijkstra's algorithm on integer \posS, then that heap can be used to create a \partialOrdering partial ordering for floating point \posS. 
\end{theorem}
\begin{proof}
By induction, it is easy to see (from the above discussion) that nodes are processed such that ${i < j \Rightarrow \actual(\node_i) <  \actual(\node_j) + \minEdgeLength}$. Thus, nodes are processed in a \partialOrdering. Combining this result with Theorem~\ref{th:correct} completes the proof.
\qed
\end{proof}

In the worst case, all nodes always require a unique integer key. Thus, worst-case runtime for the resulting float \posS algorithm is identical to that for the underlying integer \nonnegS algorithm with respect to $\numNodes$ and $\numEdges$. That said, it is easy to imagine adding additional machinery to some preexisting heaps so that all nodes with the same integer key $\key(\node_i)$ are inserted into a single bucket, and the heap sorts buckets in place of individual nodes --- an idea that may reduce the expected running time in some case, although we do not pursue it further in the current paper.

A notable difference between using an integer solution as an oracle for floating point data vs.\ using that same integer solution on raw integer data is the fact that we calculate integer keys by dividing floating point values by $\maxEdgeLength$. Thus, when used in conjunction with an integer priority queue solution that normally has a runtime dependent on $\maxEdgeLength$ (e.g., bucketing based queues), the $\maxEdgeLength$ used in the raw integer \nonnegS solution must be replaced by $\maxEdgeLength/\minEdgeLength$ in the resulting floating point \posS algorithm. This leads to the following corollary:

\begin{corollary} \label{co:po}
Assuming an integer \posS solution is used as an oracle to solve a floating point \posS, the latter will have similar runtime vs.\ the former with respect to $\numNodes$ and $\numEdges$, while  $\maxEdgeLength$ will be replaced by $\maxEdgeLength/\minEdgeLength$. 
\end{corollary}

Integer \nonnegS trivially includes all of \posS, since \nonnegS is a super-set of \posS. This leads to two additional corollaries.

\begin{corollary} \label{th:nonneg}
Assuming a heap $\heap$ exists that creates a FO when used by Dijkstra's algorithm on integer \nonnegS, then that heap can be used to create a \partialOrdering partial ordering for floating point \nonnegS. 
\end{corollary}

\begin{corollary} \label{co:nonneg}
Assuming an integer \nonnegS solution is used as an oracle to solve a floating point \posS, the latter will have similar runtime vs.\ the former with respect to $\numNodes$ and $\numEdges$, while $\maxEdgeLength$ will be replaced by $\maxEdgeLength/\minEdgeLength$. 
\end{corollary}

\section{Faster Runtimes For Floating Point \posS}

We now discuss new runtimes for floating point \posS that are enabled by Corollary~\ref{co:po}. The work by \cite{Thorup4} solves integer \nonnegS in time ${\bigOh{\numEdges + \numNodes \log \log \maxEdgeLength}}$, combining their result with corollary~\ref{co:po} gives:

\begin{corollary} \label{co:a}
Floating point \posS can be solved in time ${\bigOh{\numEdges + \numNodes \log \log \frac{\maxEdgeLength}{\minEdgeLength}  }}$.
\end{corollary}

\begin{corollary} \label{co:a.1}
Floating point \posS can be solved in expected time ${\bigOh{\numEdges + \numNodes \log \log \frac{\maxEdgeLength}{\minEdgeLength}  }}$.
\end{corollary}

Similarly, \cite{Raman2} proves that a combination of \cite{Ahuja.etal} and \cite{Cherkassky.etal} can solve integer \nonnegS in time $\bigOh{\numEdges + \numNodes(\log \maxEdgeLength \log \log \maxEdgeLength)^{1/3}}$ so:

\begin{corollary} \label{co:b}
Floating point \posS can be solved in time ${\bigOh{\numEdges + \numNodes(\log \frac{\maxEdgeLength}{\minEdgeLength} \log \log \frac{\maxEdgeLength}{\minEdgeLength})^{1/3}}}$.
\end{corollary}

While these algorithms/corollaries cover a large chunk of the problem space defined by ${\numNodes \times \numEdges \times \maxEdgeLength}$, there are a variety of additional integer \nonnegS and \posS algorithms that provide the best theoretical runtime for less frequently addressed regions of that space. We now briefly summarize these.
Leveraging integer-based result from \cite{Thorup3}:
\begin{corollary} \label{co:e}
Floating point \posS can be solved in time ${\bigOh{\numNodes + \numEdges\log\log\numEdges}}$. 
\end{corollary}
and \cite{Raman2}:
\begin{corollary} \label{co:f}
Floating point \posS can be solved in time ${\bigOh{\numNodes + \numEdges\log\log \frac{\maxEdgeLength}{\minEdgeLength}}}$. 
\end{corollary}
and \cite{Cherkassky.etal.MP96} for both:
\begin{corollary} \label{co:g}
Floating point \posS can be solved in time ${\bigOh{\numEdges + \numNodes(B+\frac{\maxEdgeLength}{B\minEdgeLength})}}$ for user defined parameter ${B < C+1}$
\end{corollary}
and:
\begin{corollary} \label{co:h}
Floating point \posS can be solved in time ${\bigOh{\numEdges\Delta + \numNodes(\Delta + \frac{\maxEdgeLength}{\Delta\minEdgeLength})}}$,  for user defined $\Delta$.
\end{corollary}

Corollary~\ref{co:h} is the final corollary regarding runtime presented in our current paper.

%
%

\section{Correct Dijkstra's on \nonnegS Requires a Full Ordering} \label{sec:correctfull}

It would be nice if our method could somehow be extended to handle the case ${\minEdgeLength = 0}$ in addition to ${\minEdgeLength > 0}$; however, we now prove this is impossible.
We first show that \partialOrdering is insufficient to correctly solve \nonnegS, and then proceed to prove that correctly solving \nonnegS requires a full ordering.

\subsection{Insufficiency of \partialOrdering to Solve \nonnegS using Dijkstra's Algorithm}

We have defined \partialOrdering such that it only applies to \posS.
Even if we attempt to accommodate \nonnegS by letting ${\minEdgeLength = 0}$ then \partialOrdering degenerates into ${i < j \Rightarrow \actual(\node_i) < \actual(\node_j)}$ where the latter inequality is strict. Thus, \partialOrdering can only exist if ${\edgeLength{\edge_{ij}} > 0}$ for all ${\edge_{ij} \in \edgeSet}$ --- otherwise the existence of some ${\edgeLength{\edge_{ij}} = 0}$ would guarantee that for some $i$ and $j$ it is the case that ${i < j}$ despite the fact that ${\actual(\node_i) = \actual(\node_j)}$, which would violate \partialOrdering. This contradiction leads to the following lemma:

\begin{lemma} \label{th:insuf}
\partialOrdering is insufficient to solve \nonnegS.
\end{lemma}

\subsection{A Necessary Condition for Correct Dijkstra's on \nonnegS}

Lemma~\ref{th:insuf} suggests that allowing zero-length edges fundamentally changes the nature of the problem.
However, it falls short of disproving that {\it all} best-path-length-based partial orderings (that are not also full orderings) are similarly doomed with respect to \nonnegS. 
In pursuit of the latter, we now prove a condition required for a correct Dijkstra's algorithm and then define another ordering that is related to \partialOrdering but slightly more general.

We begin by proving a necessary condition for any ordering that correctly solves \nonnegS. This particular condition is useful as an analytical tool simply because we can show that it is necessary (Theorem~\ref{lem:nec}) and we can also show that it implies a full ordering is required for floating point \nonnegS in the worst case (Theorem \ref{lem:worstcase}).
In particular, the following condition guarantees that if the best path to $\node_c$ involves a zero-length sub-path immediately before reaching $\node_c$, and also there exists a worse path to $\node_c$ via some other node $\node_b$, then all of the nodes involved in the zero-length subpath must be processed before $\node_c$ in order to correctly solve \nonnegS.\\

\noindent{\bf Condition A}.
{\it If} for all $\pathSequence^*(s, \node_c)$ it is the case that there exists some $\node_a$ and ${\pathSequence^*(\node_a, \node_c) \subset \pathSequence^*(s, \node_c)}$ such that ${\edgeLength{\pathSequence^*(\node_a, \node_c)} = 0}$, {\it then}:
\begin{equation} \label{eq:extra}
\text{if}
\begin{cases}
\actual(\node_c) = \actual(\node_a) < \actual(\node_b)\\
\edgeLength{\pathSequence^*(\node_a, \node_c)} = 0\\
\edgeLength{\pathSequence^*(s, \node_b)} + \edgeLength{\edge_{bc}} > \actual(\node_c)\\
\end{cases}
\text{ then } i \leq c \text{ for all } i \text{ s.t. } \node_i \in \pathSequence^*(\node_a, \node_c)
\end{equation}
where we technically only require that  $\pathSequence^*(\node_a, \node_c)$ is a zero-length subset of {\it any} optimal path to $\node_c$ such that ${\node_a \neq \node_c}$ and ${\actual(\node_c) = \actual(\node_a)}$. In other words, we can break ties between two or more optimal paths arbitrarily.

\begin{lemma} \label{lem:nec}
Condition A is necessary to correctly solve \nonnegS using Dijkstra's algorithm.
\end{lemma}
\begin{proof}
(by contradiction) Assume that ${c < i}$. Then by construction it is possible to create a graph satisfying Condition A such that nodes are processed in an order such that ${b < c < i}$, and consequently ${\guess(\node_c) = \edgeLength{\pathSequence^*(s, \node_b)} + \edgeLength{\edge_{bc}} \geq \actual(\node_c)}$ at the instant $\node_c$ is processed.
\qed
\end{proof}

\subsection{Another Ordering}

Let \alternativeOrdering be defined as a sequence ${v_1, \ldots, v_n}$ such that both:
\begin{itemize}
\item Condition A is satisfied
\item $i < j \Rightarrow \actual(\node_i) \leq \actual(\node_j) + \smallThing$
\end{itemize}
where $\smallThing$ is a constant that we are allowed to choose per each problem instance.

The definition of \alternativeOrdering allows  ${\actual(\node_i) = \actual(\node_j)}$ and so it can handle zero-length edges even when ${\smallThing = 0}$ (it handles them trivially when ${\smallThing > 0}$).

\subsection{For All $\smallThing \geq 0$, \alternativeOrdering Induces a Full Ordering for \nonnegS in the Worst Case}

We now prove that, for all ${\smallThing \geq 0}$, \alternativeOrdering induces a full ordering for \nonnegS in the worst case. We begin by examining \alternativeOrdering when  ${\smallThing > 0}$.

\begin{figure}[h!]
\centering

  \begin{xy}
    \xyimport(100,100){
      \includegraphics[width=7.5cm, trim=0 0 0 0, clip=true]{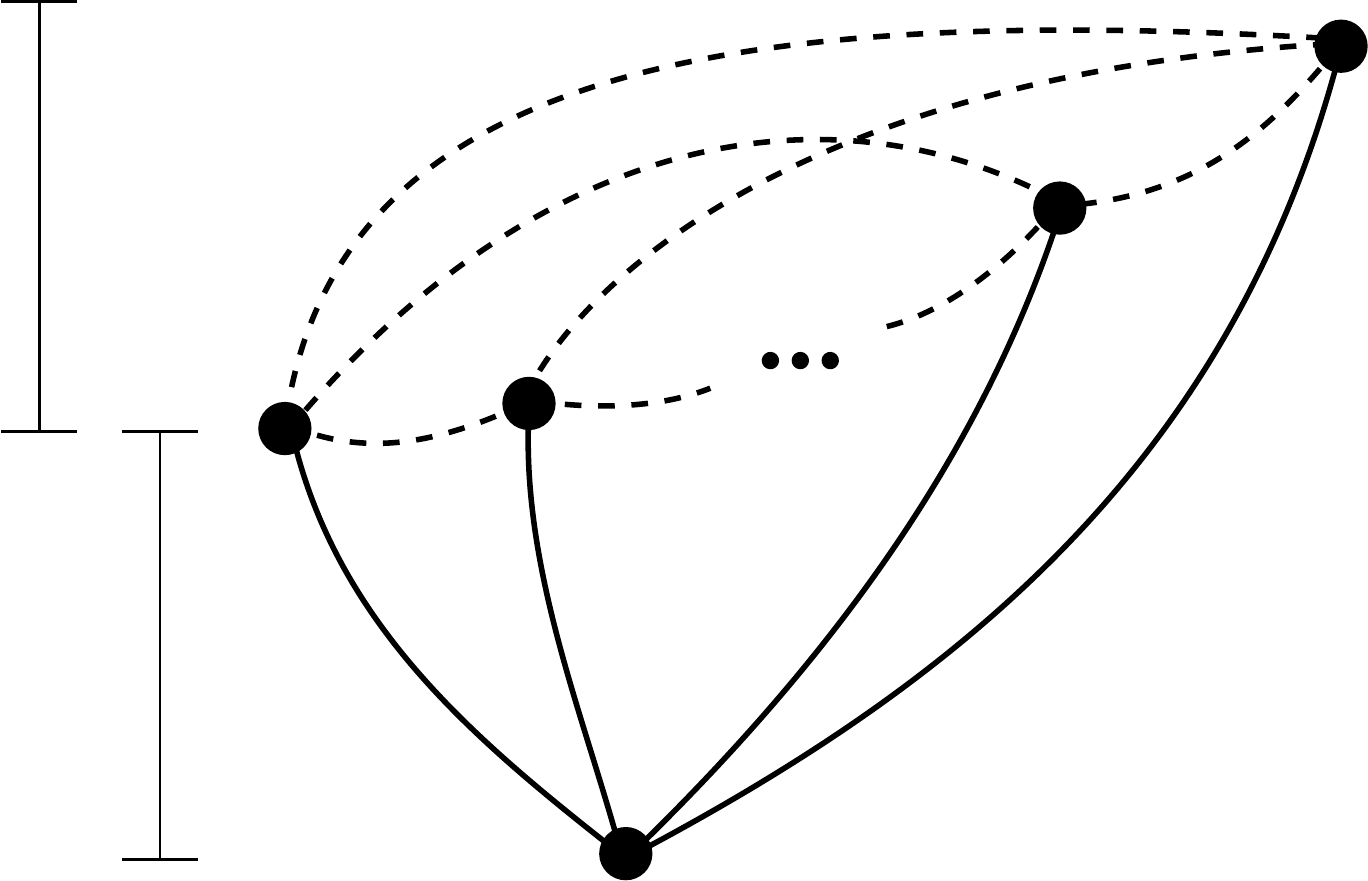}}
    ,(8,75)*{\smallThing}
    ,(6,25)*{\actual(\node_2)}
    ,(52,0)*{\node_1}
    ,(27,44)*{\node_2}
    ,(43,50)*{\node_3}
    ,(84,71)*{\node_{n-1}}
    ,(103,90)*{\node_{n}}
    ,(33,30)*{\edge_{12}}
    ,(45,35)*{\edge_{13}}
    ,(60,40)*{\edge_{1(n-1)}}
    ,(90,45)*{\edge_{1n}}
    ,(50,98)*{\edge_{2n}}
    ,(32,54)*{\edge_{23}}
  \end{xy}
  \caption{Dotted edges have zero length while the solid edges have length defined by their height span in the figure (width is used for illustrative purposes only).
For all $i$ and $j$ such that ${2 \leq i < n}$ and ${2 < j \leq n}$ and ${i < j}$ edge lengths are defined such that 
${0 < \edgeLength{\edge_{1i}} < \edgeLength{\edge_{1j}}}$ 
and
${\edgeLength{\edge_{1j}} - \edgeLength{\edge_{1i}} < \smallThing}$
and
${\edgeLength{\edge_{ij}}} = 0$ 
} \label{fig:proof}
\end{figure}

Consider Figure~\ref{fig:proof}, the dotted edges have zero length while the solid edges have length defined by their height span in the figure (width is used for illustrative purposes only). By construction, for all $i$ and $j$ such that ${2 \leq i < n}$ and ${2 < j \leq n}$ and ${i < j}$ edge lengths are defined such that 
${0 < \edgeLength{\edge_{1i}} < \edgeLength{\edge_{1j}}}$ 
and
${\edgeLength{\edge_{1j}} - \edgeLength{\edge_{1i}} < \smallThing}$
and
${\edgeLength{\edge_{ij}}} = 0$ 
Therefore, ${\actual(\node_k) = \actual(\node_2)}$ for all $k$ such that ${1 < k \leq n}$.

\begin{lemma} \label{lem:worstcase}
Condition A induces a full ordering in the worst case for any $\smallThing > 0$
\end{lemma}
\begin{proof}
(by construction) For each $k$ such that ${2 < k < n}$ there exists
${\pathSequence^*(\node_{k-1}, \node_k) \subset \pathSequence^*(s, \node_k)}$ such that ${\edgeLength{\pathSequence^*(\node_{k-1}, \node_1)} = 0}$ 
and
${\actual(\node_k) = \actual(\node_2) < \actual(\node_n)}$ 
and
${\edgeLength{\pathSequence^*(\node_n, \node_k)} = 0}$
and
${\edgeLength{\pathSequence^*(s, \node_n)} + \edgeLength{\edge_{nk}} = \actual{\node_n} > \actual(\node_k)}$. Thus, Condition A requires that ${{k-1} \leq k < n}$. Trivially we know that ${\startNode = \node_1}$ and so ${1 < k}$, thus there is only one possible extraction order that will enable Dijkstra's algorithm to yield a correct solution to this \nonnegS: ${1, 2, \ldots, n}$ --- and this is a full ordering. Moreover, this can be accomplished for any $\smallThing > 0$.
\qed
\end{proof}

\noindent We are now ready to consider ${\smallThing \geq 0}$.

\begin{theorem} \label{th:full}
Assuming that Dijkstra's algorithm processes nodes according to a partial ordering ${i < j \Rightarrow \actual(\node_i) \leq \actual(\node_j) + \smallThing}$, where ${\smallThing \geq 0}$, any partial ordering that yields a correct algorithm will induce a full ordering in the worst case.
\end{theorem}
\begin{proof}
There are two cases: $\epsilon > 0$ and ${\epsilon = 0}$.
In the first case, we have ${\epsilon > 0}$ and Lemmas~\ref{lem:nec} and \ref{lem:worstcase} guarantee that in the worst case we will create a full ordering. In the second case, we have ${\epsilon = 0}$ and so we recover the definition of a full ordering, i.e., ${i < j \Rightarrow \actual(\node_i) \leq \actual(\node_j)}$.
\qed
\end{proof}

\section{Contributions and Related Work} \label{sec:relatedwork}

Our primary problem of interest is the {\it floating point} version of \posS.  The primary result of this paper is a general technique that allows any monotonic integer priority queue to be used to solve floating point \posS, as well as understanding why this is possible and where the technique breaks down. 

Our contributions regarding the theoretical runtime of Dijkstra's algorithm can most easily be seen by comparing Tables~\ref{table:before} and \ref{table:after}.  {\it It is important to note that our new runtime results are only applicable to cases involving directed graphs with positive floating point weights}, i.e., floating point \posS.

For disconnected graphs it is possible that  ${\numEdges < \numNodes}$, and so we report runtimes in terms of both $\numEdges$ and $\numNodes$ (other authors have occasionally omitted terms involving only $\numNodes$ because ${\numNodes = \bigOh{\numEdges}}$ for {\it connected} graphs). 
A technical report that includes some of the ideas in this work appears in \cite{otte.tr15}; however, \cite{otte.tr15} focuses on a few select applications of the general idea presented in the current paper.

\subsection{Comparison Computational Model}

Dijkstra's algorithm for \nonnegS was originally presented in \cite{Dijkstra} with a runtime of ${\bigOh{\numEdges + \numNodes^2}}$ and assuming a comparison based computational model.
\cite{Williams} refined this to $\bigOh{\numEdges\log(\numNodes) + \numNodes}$ using a more sophisticated heap. 
The Fibonacci heap by \cite{Fredman.Tarjan} allows the (integer and floating point) \nonnegS to be solved in ${\bigOh{\numEdges + \numNodes\log\numNodes}}$ --- this is believed to be the fastest possible theoretical runtime for a comparison-based computational model.

\subsection{RAM Computational Model}

An alternative style of algorithm assumes a RAM computational model in order to use heaps based on some form of bucketing. 
%
%
\cite{Dial.ACM69} uses an approximate heap data structure to achieve runtime ${\bigOh{\numEdges + \maxEdgeLength\numNodes}}$ for integer \nonnegBoundedS. 
\cite{Ahuja.etal} proves that a combination of a Radix-heap and a Fibonacci-heap will solve integer \nonnegBoundedS in ${\bigOh{\numEdges + \numNodes\sqrt{\log \maxEdgeLength}}}$.
\cite{Fredman.Willard} use AF-heaps to achieve ${\bigOh{\numEdges + \numNodes \frac{\log \numNodes}{\log \log \numNodes} }}$. Note that AF-heaps rely on Q-heaps, which assume that edge weights are integers.
The ideas of \cite{Dial.ACM69} are extended by \cite{tsitsiklis.95} to floating point \posBoundedS in a RAM model with a runtime of ${\bigOh{\numNodes + \numEdges + \maxEdgeLength/\minEdgeLength}}$.

For integer \nonnegBoundedS in the RAM model, \cite{Cherkassky.etal.MP96} uses a more complex variation of \cite{Dial.ACM69}  to achieve time ${\bigOh{\numEdges + \numNodes(B+C/B)}}$ and time ${\bigOh{\numEdges\Delta + \numNodes(\Delta + C/\Delta)}}$, for user defined parameters ${B < C+1}$ and $\Delta$.
\cite{Thorup3} achieves ${\bigOh{\numNodes + \numEdges\log\log\numEdges}}$.
\cite{Raman2} shows that combining results from \cite{Ahuja.etal} with those from \cite{Cherkassky.etal.MP96} yields a ${\bigOh{\numEdges + \numNodes(\log \maxEdgeLength \log \log \maxEdgeLength)^{1/3}}}$  solution to integer \nonnegBoundedS, and that \cite{Thorup3} can be extended to achieve ${\bigOh{\numNodes + \numEdges\log\log \maxEdgeLength }}$.
A later solution by \cite{Cherkassky.etal} achieves ${\bigOh{\numEdges + \numNodes \frac{\log \maxEdgeLength}{\log \log \maxEdgeLength} }}$. Although this is slightly worse than the result by \cite{Raman2}, it is of interest because the hot-queue data structure it uses can be modified to solve floating point \posBoundedS in ${\bigOh{\numEdges + \numNodes \frac{\log \frac{\maxEdgeLength}{\minEdgeLength}}{\log \log \frac{\maxEdgeLength}{\minEdgeLength}}}}$ by incorporating the earlier ideas from \cite{tsitsiklis.95}.

\cite{Thorup4}  presents an integer priority queue that has constant time for the decrease key operation and is thus able to achieve a runtimes of $\bigOh{\numEdges + \numNodes \log \log \numNodes}$ and $\bigOh{\numEdges + \numNodes \log \log \maxEdgeLength}$ for integer \nonnegS and \nonnegBoundedS, respectively. Using the trick of interpreting IEEE floating point numbers as integers enables the algorithm to also solve floating point \nonnegS in time $\bigOh{\numEdges + \numNodes \log \log \numNodes}$. However, as the author of \cite{Thorup4} notes, ``for floating point numbers we [they] do not get bounds in terms of the maximal weight'' (i.e., $\maxEdgeLength$). 
We note that, as proved in Corollary \ref{co:a}, our method can be combined with the integer result of $\bigOh{\numEdges + \numNodes \log \log \maxEdgeLength}$ from \cite{Thorup4} to yields a new algorithm with a runtime of $\bigOh{\numEdges + \numNodes \log \log \maxEdgeLength/\minEdgeLength}$ for floating point \posBoundedS. 
This highlights one way our work differs from previous work. A handful of other improvements that are newly obtained by combining our work with previous integer results appear in Table~\ref{table:after}.

\subsection{Expected Times in RAM Model}

With regard to {\it expected} running times over randomizations (of edge lengths) and assuming RAM and integer \nonnegS,  \cite{Fredman.Willard2} gives expected time ${\bigOh{\numEdges\sqrt{\log \numNodes } + \numNodes}}$.
More recent results  by \cite{Thorup3} yield the expected time of $\bigOh{\numEdges + \numNodes(\log \numNodes)^{1/2 + \epsilon}}$, while \cite{Raman} and \cite{Raman2} give expected times ${\bigOh{\numEdges + \numNodes(\log \numNodes \log \log \numNodes)^{1/2}}}$ and ${\bigOh{\numEdges + \numNodes (\log \maxEdgeLength)^{1/4+\epsilon} }}$ and ${\bigOh{\numEdges + \numNodes(\log \numNodes)^{1/3+\epsilon}}}$.
\cite{Cherkassky.etal} achieves the slightly worse expected time  ${\bigOh{\numEdges + \numNodes(\log \maxEdgeLength)^{1/3+\epsilon}}}$; however, we note that this hot-queues-based solution can be combined with \cite{tsitsiklis.95} to solve floating point \nonnegS in expected time ${\bigOh{\numEdges + \numNodes(\log \frac{\maxEdgeLength}{\minEdgeLength})^{1/3+\epsilon}}}$ --- so we include this in Table~\ref{table:before}.
\cite{Thorup5} proves a general reduction from \nonnegS to sorting that leverages the randomized sorting algorithm in \cite{han.thorup} to produce an integer \nonnegS algorithm that runs in expected time $\numEdges\sqrt{\log \log \numNodes}$ (this can be extended to floating point \posS using the trick of interpreting IEEE floats as integers).

The more recent work of \cite{Thorup4} provides solutions in worst case time (and therefore also expected time) $\bigOh{\numEdges + \numNodes \log \log \numNodes}$ and $\bigOh{\numEdges + \numNodes \log \log \maxEdgeLength}$ for integer \nonnegS and \nonnegBoundedS, respectively. As noted above \cite{Thorup4} can be applied directly to floating point \nonnegS (by assuming IEEE standard floats are interpreted as integers) with runtime $\bigOh{\numEdges + \numNodes \log \log \numNodes}$, but analogous floating point results in terms of $\maxEdgeLength$ could not be achieved. On the other hand, our results from the current paper can be combined with \cite{Thorup4} to yield a new algorithm with a runtime of $\bigOh{\numEdges + \numNodes \log \log \maxEdgeLength/\minEdgeLength}$ for floating point \posBoundedS.

\subsection{Undirected Graphs (BN-USSSP)}

For the case of undirected edges with bounded positive integer weights (integer BN-USSSP) \cite{Thorup} presents a method that runs in $\bigOh{\numEdges + \numNodes}$, and then extends this to floating point BN-USSSP in \cite{Thorup2}. In \cite{Thorup2} it is noted that an alternative solution with significantly less overhead is also possible in ${\bigOh{\log(\minEdgeLength) + \alpha(\numEdges,\numNodes)\numEdges + \numNodes}}$, where $\alpha(\numEdges,\numNodes)$ is the inverse Ackermann function using $\numEdges$ and $\numNodes$.
\cite{wei.tanaka,wei.tanaka2} present modifications to \cite{Thorup2} that achieve better practical performance by removing the necessity of an ``unvisited node structure'' but do not change theoretical runtime bounds (or the constant factor). 

\subsection{Other Special Cases}

The special case of (integer or floating point) \posS involving $z$ distinct weights can be solved in time $\bigOh{\numEdges + \numNodes}$ if ${z\numNodes \leq 2\numEdges}$ and $\bigOh{\numEdges\log\frac{z\numNodes}{\numEdges} + \numNodes}$, otherwise, using  \cite{Orlin.etal}. 
Integer or floating point \nonnegS on planar graphs can be solved in $\bigOh{\numNodes\sqrt{ \numNodes}}$ with a method by \cite{Federickson}, note that ${\numEdges = \bigOh{\numNodes}}$ for planar graphs.

\subsection{Partial Orderings}

The main conceptual difference between \cite{Thorup,Thorup2} and our work is in the details of the partial orderings that are used. These differences cause both (A) theoretical ramifications regarding the subsets of SSSP for which a particular is applicable, and (B) practical differences affecting ease of implementation and performance. The particular partial ordering we investigate \partialOrdering is grounded in the notion of shortest-path-length.
Our method is designed for floating point \posS, while \cite{Thorup,Thorup2} is for integer and floating point BN-USSSP. The floating point method in \cite{Thorup2} uses the integer algorithm in \cite{Thorup} as an oracle; our work extends the basic oracle concept to a greater subset of SSSP using a completely different technique.

The observation that \posS can be solved using \partialOrdering was first observed by \cite{dinic.78}, who used the trick of dividing by minimum edge weight to improve the use of bucket based integer priority queues for solving integer SSSP. \cite{tsitsiklis.95} proves the result for floating point \posS and uses it in a variation of \cite{Dial.ACM69} to solve floating point \posBoundedS; although, without explicitly considering the runtime of the resulting algorithm with respect to $\maxEdgeLength$ or $\minEdgeLength$. 
Our work can be also be considered an application of \cite{dinic.78} and a generalization of \cite{tsitsiklis.95}. We show that any monotone integer priority queue that creates a \fullOrdering for integer \nonnegS (or \posS) can be used as an oracle to create a \partialOrdering for floating point \posS, and thereby solve floating point \posS. Finally, we prove that \partialOrdering is insufficient to solve floating point \nonnegS and that any related ordering that correctly solves floating point \posS must create a \fullOrdering in the worst-case.

\section{Remarks}

We now discuss a few points that we find interesting.

\subsection{Integer \nonnegS vs.\ Floating Point \nonnegS}

Even though we have shown that heaps designed to solve integer \nonnegS can immediately be used to solve floating point \posS in general, we fall short of achieving a similarly general oracle method that can solve floating point \nonnegS.
The reason for this discrepancy is intimately related to zero-length edges and enforcement of Condition A.  It is impossible to make the granularity of the partial ordering small enough (i.e.,  without creating full ordering) to guarantee that processing nodes based on ${\key(\node) = \guess(\node)/\minEdgeLength}$ instead of $\guess(\node)$ yields a correct algorithm.
Indeed, this result extends to any method relying on ${\key(\node) = \guess(\node)/x}$, regardless of how $x$ is chosen. There will always be problems when ${x = 0}$. 

More interesting is the result from Lemma~\ref{lem:nec}, which proves that Condition A is necessary for any ordering that correctly solves \nonnegS. 
Note that this does not conflict with previous results for integer \nonnegS case --- it is trivial to show that full orderings guarantee Condition A, and all state-of-the-art integer \nonnegS algorithms based on shortest-path-length partial orders use full orderings. 
Indeed, the main ramification of Lemma~\ref{lem:nec} appears to be that it is impossible for any version of Dijkstra's algorithm to solve general-case floating point \nonnegS by using a (non-full) partial ordering based on shortest-path length. Thus, future work on the general floating point \nonnegS problem must either: (1) use an alternative method that in which orderings consider more than just best-path-length (as \cite{Thorup} has done for the USSSP case), or (2) attempt to create faster full-order inducing heaps that use floating point keys directly. 

\subsection{Floats, Ints and the IEEE Standard}

It has often been remarked, e.g., by \cite{Thorup2}, that given the IEEE standards for floating point and integer representations it is often possible to use floating point numbers directly in some integer heaps by reinterpreting their binary values directly as integers.
This is obviously true for comparison based algorithms, but is even true for some of the bucket-based structures such as Radix heaps.

In contrast, the idea presented in the current paper is representation agnostic. This is theoretically important because there are a variety of ways to represent real numbers for computational purposes other than IEEE standard, for example, unnormalized versions of binary and decimal floating point numbers and more exotic alternatives such as: continued fractions \cite{vuillemin1990exact}, logarithmic \cite{swartzlander1975sign}, semi-logarithmic \cite{muller1998semi}, level-index \cite{clenshaw1984beyond}, floating-slash \cite{matula1985finite}, and others \cite{vuillemin1994circuits,muller2009handbook}.
Our results hold for {\it all} binary representations of positive real numbers such that an $\bigOh{1}$ integer conversion operation exists, and an integer-based heap exists that can solve \nonnegS or \posS for the resulting integer representation. 


\section{Summary and Conclusions}

We provide a new proof that using Dijkstra's algorithm directly with partial ordering based on shortest-edge-length and best-path-length is sufficient to solve the floating point  positive edge weight case of the single source shortest path planning problem (SSSP) --- we note this result has been previously obtained by \cite{dinic.78} and \cite{tsitsiklis.95} using alternative analysis.

Based on this result, we present a simple yet general method that enables a large class of results from the nonnegative (and positive) integer case of SSSP to be extended to the positive floating point case of SSSP.
In particular, any integer priority queue that solves integer SSSP by creating a full ordering (by extracting nodes in a monotonically nondecreasing sequence of best-path-length key values) can be used as an oracle to solve the floating point SSSP with positive edge weights (by creating a sufficient partial ordering).
This immediately yields a handful of faster theoretical runtimes for the positive floating point case of SSSP --- both worst case and expected times and for various relationships between $\numNodes$, $\numEdges$, and $\maxEdgeLength$. The idea is easy to implement in practice, i.e., in conjunction with many practical heaps that solve integer SSSP. Moreover, it guarantees that many future advances for the integer case, both theoretical and practical, will immediately yield better results for the positive floating point case of SSSP.

Finally, we prove that Dijkstra's Algorithm must use a full ordering to solve the nonnegative floating point of SSSP in the worst case (and thus our results cannot be extended to {\it nonnegative} floating point SSSP). The latter is due to complications inherent with using zero-length edges, and suggests that future work for the floating point case of SSSP must either involve a departure from Dijkstra's Algorithm or the discovery of heaps that are able to create full orderings over floating point numbers more quickly.

\newcommand{\removeThis}[1]{}

\removeThis{

\section*{Appendix: An anonymous review of an earlier versions of this work and our response to it}

An earlier version of this manuscript (which we shall call V1) was submitted to a top computer science journal on August 11, 2015 but was not accepted. To our discontent, we received a terse review from a single reviewer and were not given the opportunity to respond. That review is now included, in its entirety, followed by our open response to it. The current version of our paper (V2) has already been updated to address the shortcomings that were highlighted in the review of V1. In particular, we have (1) incorporated a missing reference to \cite{Thorup4} and (2) better highlighted the contributions of our work. We hope that by including the review of V1 and our response to it, we will be able to satisfy any readers that have similar concerns to our anonymous reviewer.

\subsection{A review of the first version of this paper}

\hspace{1cm} \begin{minipage}{0.8\textwidth}

\it 

``The paper considers the single source shortest path problem in
directed graphs with integer or floating point weights. The paper has no new fundamental contributions. It is in particular missing the strongest current bounds on the problem from: \cite{Thorup4}.
%
%
With m edges, n nodes, and max integer weight $\maxEdgeLength$, Thorup solves the problem in $\bigOh{\numEdges+\numNodes \log \log \min(\maxEdgeLength,\numNodes)}$ time. It is a standard idea, going back to \cite{Dial.ACM69} that we can replace $\maxEdgeLength$ with $\maxEdgeLength/\minEdgeLength$ if $\minEdgeLength$ i the minimal weight.''

\end{minipage}\\

\noindent Note that we have changed the notation from the original review to reflect the notation that is used in the current paper (V2). We have also replaced the long-form citations included by the reviewer with the references to the bibliography of V2.

\subsection{Our open response to the preceding review}

First of all, we would like to thank the anonymous reviewer for the time they invested in reading our paper (V1) and for their comments. The comments provided by the reviewer have enabled us to improve and clarify our work in the hopes of resubmitting the improved version (V2) to another venue. We now respond to each of the reviewer's comments in turn. \\ 

\noindent {\bf Reviewer Comment 1:} {\it ``The paper considers the single source shortest path problem in
directed graphs with integer or floating point weights.''}\\

Our primary problem of interest is the {\it floating point} version of SSSP. We did not highlight this fact strongly enough in V1. Our focus on floating point SSSP has been clarified in V2 (in the abstract, introduction, Section~\ref{sec:heaps_orical}, related work, and conclusions). A new high-level overview of our work now appears in Section 1.1.\\

\noindent {\bf Reviewer Comment 2:} {\it ``The paper has no new fundamental contributions. It is in particular missing the strongest current bounds on the problem from \cite{Thorup4}. With $\numEdges$ edges, $\numNodes$ nodes, and max integer weight $\maxEdgeLength$, Thorup solves the problem in $\bigOh{\numEdges+\numNodes \log \log \min(\maxEdgeLength,\numNodes)}$ time.'' }\\

This comment can be divided into three different parts related to {\bf (A)} missing related work, {\bf (B)} the fundamental contributions of our work, and {\bf (C)} new runtimes for SSSP. We now address each of these individually:\\

{\bf (A)}
The reviewer is correct that in V1 we were unaware of a related work that is of critical importance, namely \cite{Thorup4}. The reviewer is also correct that 4 of the 9 ``new'' runtime bounds we reported in V1 had already been superseded by results from \cite{Thorup4} (and 6 of 10 other results we reported for earlier algorithms had also been improved upon by \cite{Thorup4}). V2 has been updated to reflect the work of \cite{Thorup4}; see for example where the work of \cite{Thorup4} now appears in tables 1 and 2, as well as the discussion of related work (Section~\ref{sec:relatedwork}).\\

{\bf (B)}
Our paper is fundamentally about presenting a general technique that allows integer priority queues to be used to solve floating point SSSP, as well as understanding why this is possible and where the technique breaks down. 
Improving runtimes for floating point SSSP is both a nice bonus that highlights how our work can be applied, and also a strong argument for the novelty of our work. However, the particular runtimes that are achieved should not themselves be considered the main point of our work. As mentioned in the response to Reviewer Comment 1, we have added additional discussion to V2 to better highlight the main point of our paper, see for example, Section 1.1. We have also added a discussion on monotonicity in priority queues which explores a few of the finer points of our work from a slightly different angle (see Section 2.3).\\

{\bf (C)}
In both V1 and V2, we only claim(ed) improvements for {\it floating point} SSSP.  The new results from \cite{Thorup4} leave intact 5 of the 9 improvements we originally reported in V1. The combination of our work with \cite{Thorup4} actually yields 2 new runtime improvements that are now presented in V2. In particular, a new algorithm for floating point SSSP with positive edge weights that has runtime (both worst-case and expected) of ${\bigOh{\numEdges + \numNodes \log \log \frac{\maxEdgeLength}{\minEdgeLength}  }}$. We note that in \cite{Thorup4} the author explicitly stated he was unable to find a floating point bound in terms of $\maxEdgeLength$, saying ``for floating point numbers we [they] do not get bounds in terms of the maximal weight,'' i.e., $\maxEdgeLength$. 

Other faster runtimes for floating point SSSP that our work enables include:
${\bigOh{\numEdges + \numNodes(\log \frac{\maxEdgeLength}{\minEdgeLength} \log \log \frac{\maxEdgeLength}{\minEdgeLength})^{1/3}}}$, which is asymptotically faster than ${\bigOh{\numEdges + \numNodes \log \log \frac{\maxEdgeLength}{\minEdgeLength}  }}$ for many values of $\maxEdgeLength$; ${\bigOh{\numNodes + \numEdges\log\log\numEdges}}$ and ${\bigOh{\numNodes + \numEdges \log \log \frac{\maxEdgeLength}{\minEdgeLength}  }}$, which are relevant when ${\numEdges \leq \numNodes}$; and  ${\bigOh{\numEdges + \numNodes(B+\frac{\maxEdgeLength}{B\minEdgeLength})}}$ and ${\bigOh{\numEdges\Delta + \numNodes(\Delta + \frac{\maxEdgeLength}{\Delta\minEdgeLength})}}$, which involve user parameters.

Finally, although we are unable to surpass the bound of  $\bigOh{\numEdges + \numNodes \log \log \numNodes}$ reported in \cite{Thorup4}, we {\it are} able to match it (for graphs with positive edge weights) using a fundamentally different technique than the one presented in \cite{Thorup4}. We note that in V1 we did {\it not} claim to achieve the fasted bound in terms of only $\numNodes$ and $\numEdges$.\\

{\bf Reviewer Comment 3:}  {\it ``It is a standard idea, going back to \cite{Dial.ACM69} that we can replace $\maxEdgeLength$ with $\maxEdgeLength/\minEdgeLength$ if $\minEdgeLength$ i the minimal weight.''}\\

Our response to this comment assumes that the reviewer meant ``we can replace $\maxEdgeLength$ with $\maxEdgeLength/\minEdgeLength$ if $\minEdgeLength$ {\it is} the minimal weight'' (i.e., that `i' is not an index on $\minEdgeLength$).

We agree that replacing the maximum edge length $\maxEdgeLength$ with $\maxEdgeLength/\minEdgeLength$, where $\minEdgeLength$ is the minimum edge length, is a well known idea that goes back to \cite{Dial.ACM69}.  Indeed, in both V1 and V2, we explicitly state, e.g.,  in our abstract, that ``The method relies on a result known to \cite{dinic.78} and \cite{tsitsiklis.95}; that the single-source shortest path planning problem (SSSP) can be solved using a partial ordering of nodes (despite that full orderings are typically used in practice).'' Similar statements also appear(ed) in Section 1.1 (V2) and 2.2 (V1 and V2). 

It is clear that in V1 we failed to convey that our results go beyond simply using $\maxEdgeLength/\minEdgeLength$ instead of $\maxEdgeLength$ for the analysis of floating point SSSP algorithms based on Dial's work (as is done by [Tsitsiklis 1995]). And also that they go beyond using $\lfloor\frac{\guess(\node)}{\minEdgeLength} \rfloor$ instead of $\guess(\node)$ within bucketing algorithms, in general, as was shown by \cite{dinic.78}.

Prior to our work, it was not well understood that {\it any} monotone {\it integer} priority queue could be used to produce a {\it floating point} priority queue sufficient to solve {\it floating point} SSSP with positive edge weights.  Consider the fact that Thorup (an undisputed expert on SSSP) remarks in \cite{Thorup4} that ``for floating point numbers we [they] do not get bounds in terms of the maximal weight'' (i.e., $\maxEdgeLength$). We can assume that Thorup was familiar with Dinic's insight into partial orderings given Thorup's detailed remarks on \cite{dinic.78} in the earlier \cite{Thorup}. 
The fact that Thorup, a word-class expert with a good understanding of Dinic's ideas, was unable to achieve the bound we find in the current paper is strong evidence that our work is providing something new beyond \cite{dinic.78}.
In particular, there are two additional insights needed to leverage Dinic's partial orders to achieve floating point results:
\begin{enumerate}
\item Dinic's idea can be extended from bucket based priority queues to {\it any} monotonic priority queue.
\item By doing so it is possible to turn any monotonic integer priority queue into a floating point priority queue that can solve floating point SSSP with positive edge weights. 
\end{enumerate}
Our paper formalizes and studies these two insights.

}

\section*{Acknowledgments}

This work was supported by the Control Science Center of Excellence at the Air Force Research Laboratory (AFRL), the National Science Foundation (NSF) grant IIP-1161029, and the Center for Unmanned Aircraft Systems. The majority of this work was completed while the author was ``in residence'' at AFRL. 

\bibliographystyle{apalike}
\bibliography{bibliography}{}


\end{document}